\documentclass[12pt]{amsart}
\usepackage{mathrsfs}
\usepackage{deleq}
\usepackage{epsfig}
\usepackage{amssymb}
\usepackage{amsfonts}
\usepackage{latexsym}
\usepackage{graphics}
\usepackage{graphicx}
\usepackage{amsmath}
\usepackage{hyperref}
\usepackage{amsfonts} \usepackage{amsthm}

\newtheorem{Proposition}{Proposition}
\newtheorem{Remark}[Proposition]{Remark}
\newtheorem{Corollary}[Proposition]{Corollary}
\newtheorem{Lemma}[Proposition]{Lemma}

\newtheorem{Theorem}[Proposition]{Theorem}

\makeatletter
\@addtoreset{equation}{section} \makeatother
\def\z{\noindent}

\def\z{\noindent}

\def\sqr#1#2{{\vcenter{\vbox{\hrule height .#2pt
\hbox{\vrule width .#2pt height#1pt \kern#1pt
\vrule width .#2pt}
\hrule height .#2pt}}}}

\def\NN{\mathbb{N}}

\def\RR{\mathbb{R}}
\def\ZZ{\mathbb{Z}}

\def\iB{\mathcal {B}}

\begin{document}

\title[Ionization in damped time-harmonic fields]{Ionization in damped time-harmonic fields}
\author{O. Costin, M. Huang, Z. Qiu}

\
\maketitle

\bigskip

\begin{abstract}
  We study the asymptotic behavior of the wave function in a simple
  one dimensional model of ionization by pulses, in which the
  time-dependent potential is of the form
  $V(x,t)=-2\delta(x)(1-e^{-\lambda t} \cos\omega t)$, where $\delta$
  is the Dirac distribution.

  We find the ionization probability in the limit $t\to\infty$ for all
  $\lambda$ and $\omega$. The long pulse limit is very singular, and,
  for $\omega=0$, the survival probability is $const\,\lambda^{1/3}$,
  much larger than  $O(\lambda)$, the one in the  abrupt transition counterpart,
$V(x,t)=\delta(x)\mathbf{1}_{\{t\ge 1/\lambda\}}$
   where $\mathbf{1}$ is the
  Heaviside function.
\bigskip

\end{abstract}

\section{Introduction}
Quantum systems subjected to external time-periodic fields which are not small
have been studied in various settings.

In constant amplitude small enough oscillating fields, perturbation
theory typically applies and ionization is generic (the probability of
finding the particle in any bounded region vanishes as time becomes
large).

For larger time-periodic fields, a number of rigorous results have
been recently obtained, see \cite{Hydrogen} and references therein,
showing generic ionization. However, outside perturbation theory,
the systems show a very complex, and often nonintuitive behavior. The
ionization fraction at a given time is not always monotonic with the field
\cite{UAB}.  There even exist exceptional potentials of the form
$\delta(x)(1+a F(t))$ with $F$ periodic and of zero average, for which ionization occurs
for all small $a$, while at larger fields the particle becomes
confined once again \cite{CMP}.  Furthermore, if $\delta(x)$  is replaced with smooth potentials $f_n$ such that $f_n\to\delta$
in distributions, then ionization occurs for all $a$ if
$n$ is kept fixed.

 Numerical approaches are very delicate since one deals with the Schr\"odinger
equation in $\RR^n\times \RR^+$, as $t\to \infty$ and
artefacts such as reflections from the walls of a large box approximating the
infinite domain are not easily suppressed.  The mathematical study of systems in
various limits is delicate and important.

In physical experiments one deals with forcing of finite effective duration, often with
exponential damping.
This is the setting we study in the present paper, in a simple model,  a delta function in one
dimension,  interacting with a damped
time-harmonic external forcing. The equation is
\begin{equation}
  \label{eq:eqa} i\, \frac{\partial\psi }{\partial t}\,
  =\Big(-\frac{\partial^2}{\partial x^2} -
  2\delta(x)\left(1-A(t) \cos(\omega t)\right) \Big)\, \psi
\end{equation}
where  $A(t)$  is the amplitude of the oscillation; we take
\begin{equation}
  \label{eq:eqF}
  \psi_{0}=\psi(0,x)\in C_{0}^{\infty};\ \ A(t)=\alpha e^{-\lambda t};\ \ \alpha=1
\end{equation}
 (The analysis for other values of $\alpha$ is very similar.)
The quantity of interest is the large $t$ behavior of $\psi$, and
in particular the survival probability
\begin{equation}
    \label{eq:ioniz}
P_B=\lim_{t\to\infty} P(t,{B})=   \lim_{t\to\infty}\int_{B}|\psi(t,x)|^2dx
  \end{equation}
 where $B$ is a bounded subset of $\RR$.

{\em Perturbation theory, Fermi Golden Rule}. If $\alpha$ is small enough,
$P$ decreases exponentially on an intermediate time scale, long
enough so that by the time the behavior is not exponential anymore, the survival
probability is too low to be of physical interest. For all practical
purposes, if $\alpha$ is
small enough,
the decay is exponential, following the Fermi Golden Rule, the derivation
of which can be found in most quantum mechanics textbooks; the quantities of interest can be obtained by perturbation expansions in $\alpha$. This setting
is well understood; we mainly focus on the case where $\alpha$ is not too
small,  a toy-model of an atom interacting with
a  field comparable to the binding potential.

{\em No damping.} The case $\lambda=0$ is well understood for the
model (\ref{eq:eqa}) in all ranges of $\alpha$, see \cite{Rokhl}. In
that case, $P(t,A)\sim t^{-3}$ as $t\to\infty$.

However, since the limit $\lambda\to 0$ is singular, little
information can be drawn from the $\lambda=0$ case.

For instance, if $\omega=0$, the limiting
value of $P$ is of order $\lambda^{1/3}$, while
with an abrupt cutoff, $A(t)=\mathbf{1}_{\{t:t\le 1/\lambda\}}$, the limiting $P$ is $O(\lambda)$ (as usual, $\mathbf{1}_S$ is the characteristic
function of the set $S$).

Thus, at least for fields which are not very small, the shape of the
pulse cut-off is important. Even the simple system (\ref{eq:eqa})
exhibits a highly complex behavior.

We obtain a rapidly convergent expansion of the wave function and
the ionization probability for any frequency and amplitude; this can
be conveniently used to calculate the wave function with rigorous
bounds on errors, when the exponential decay rate is not extremely
large or small, and the amplitude is not very large. For some
relevant values of the parameters we plot the ionization fraction
as a function of time.

We also show that for $\omega=0$ the equation
is solvable in closed form, one of the few nontrivial integrable examples of the time-dependent Schr\"odinger
equation.

\section{Main results} \label{sec:setting_and_results}
\begin{Theorem} \label{thm:1}
Let $\psi(t,x)$ be the solution of  (\ref{eq:eqa}) with initial condition $\psi_{0}\in C_{0}^{\infty}$. Let
\begin{equation} \label{eq:defOfgmn}
    g_{m,n} = g_{m,n}(\sigma) = \frac{i}{2}\, \int_{\RR}  e^{-\sqrt{\sigma+n\omega-im\lambda}|x'|} \psi_0(x')dx'
\end{equation}
Then as $t\rightarrow\infty$ we have
\smallskip

\begin{equation}
    \psi(t,x) = r(\lambda,\omega)e^{i t} e^{-|x|}\left(1+t^{-1/2}h(t,x) \right)
\end{equation}
where $|h(t,x)| \leq C$, $\forall x\in \RR$, $\forall t\in \RR^{+}$, and where
\begin{equation} \label{eq:rInA}
    r(\lambda,\omega)=\left[ -A_{1,-1}-A_{1,1}+2g_{0,0}\right]_{\sigma=1}
\end{equation}
\smallskip

\z where $A_{m,n}=A_{m,n}(\sigma)$ solves
\begin{equation} \label{eq:inhomoInA1}
    (\sqrt{\sigma + n\omega - im\lambda}-1)A_{m,n}
    = -\frac{1}{2} A_{m+1,n+1} - \frac{1}{2} A_{m+1,n-1} + g_{m,n}
\end{equation}
 There is a unique solution  of (\ref{eq:inhomoInA1})
satisfying
\begin{equation}\label{eq:norm_in_A_inhomo}\sum_{m,n}(1+|n|)^{\frac{3}{2}}
e^{-b\sqrt{1+|m|}} |A_{m,n}| < \infty
\end{equation} where $b>1$ is a
constant. It is this solution that enters (\ref{eq:rInA}).

\end{Theorem}

There is a rapidly convergent representation of $r(\lambda, \omega)$, see \S \ref{subsec:sumRepOfAmn}.

Clearly, $|r(\lambda,\omega) |^{2}$ is the probability of survival,
the projection onto the limiting bound state.
\subsection{$\omega=0$}
 \begin{Theorem} \label{thm:3}
    (i) For $\omega=0$ we have
    \begin{equation}  \label{eq:inhomogeneousMainsys2}
        \begin{split}
            r(\lambda) = \int_{0}^{\infty} \frac{-e^{-p}}{1+e^{-p}} \int_{c-i\infty}^{c+i\infty} g(k) \exp \Big(\frac{2\,i\,\sqrt{-i\,k}}{\sqrt{\lambda}} \Big) \lambda^{\frac{1-k}{2}} \frac{1} {\sqrt{\Gamma(k)}} \cdot \\
            \exp \Big( -\int_{0}^{\infty} e^{-kp} \frac{\sqrt{i\,} \lambda \left( -2 + 2 e^{-p}   - i^{3/2}\sqrt{p\, \pi\,\lambda}\,{\rm{erf}}\,(\frac{-i^{3/2}\sqrt{p}} {\sqrt{\lambda}})\right)}  {2\,(-1+e^{-p})\,\sqrt{\pi}\,(p \lambda)^{3/2}} dp \Big) dk dp\\
        \end{split}
    \end{equation}
where $g(k)=g_{k,0}$.

(ii) We look at the case when $\psi_{0}=e^{-|x|}$, the bound state of the limiting
time-independent system. Assuming the series of $r(\lambda)$ is Borel
summable in $\lambda$ for $\arg \lambda \in [0, \frac{\pi}{2}] $
(summability follows from (\ref{eq:inhomogeneousMainsys2}), but the
proof is cumbersome and we omit it), as $\lambda \to 0$ we have
    \begin{equation} \label{eq:inhomogeneousMainsy1}
        r(\lambda) \sim 2^{-2/3}(-3i)^{1/6}\pi^{-1/2}\Gamma(2/3)e^{-\frac{3i}{2\lambda}}\,\lambda^{1/6}
    \end{equation}
\end{Theorem}

Note: The behavior (\ref{eq:inhomogeneousMainsy1}) is confirmed numerically with high accuracy, constants included, see \S \ref{subsec:numericalEvidence}.

We also discuss  results in two limiting cases: the short pulse setting (see \S \ref{sec:shortPulse}) and the special case $\lambda=0$ (see \S \ref{sec:anotherSpecialCase}).

\section{Proofs and further results} \label{sec:mainStrategy}

\subsection{The associated  Laplace space equation}

Existence of a strongly continuous unitary propagator for (\ref{eq:eqa})
(see \cite{Reed-Simon} v.2, Theorem X.71) implies that for $\psi_0\in
L^2(\RR^d)$, the Laplace transform
\begin{displaymath}
  \label{eq:Lap}
  \hat{\psi}(\cdot,p):=\int_0^{\infty}\psi(\cdot,t)e^{-pt}dt
\end{displaymath}
exists for $\Re(p)>0$ and the map $p\to \psi(\cdot,p)$ is $L^2$ valued
analytic in the right half plane \begin{displaymath} \label{H}
  p\in\mathbb{H}=\{z:\Re(z)>0\}
\end{displaymath}
The  Laplace transform of (\ref{eq:eqa}) is
\begin{equation} \begin{split}
  \label{eq:inhomo1} \left(\frac{\partial^2}{\partial x^2} +
    ip\right)\hat{\psi}(x,p)=i\psi_0 -
  2\delta(x)\hat{\psi}(x,p)\\
   + \delta(x)\left( \hat{\psi}(x,p - i\omega + \lambda) + \hat{\psi}(x,p + i\omega +
  \lambda) \right)
\end{split} \end{equation}

Let $p=i \sigma+m\lambda+i n\omega$ and
\begin{equation} \label{eq:yandpsi}
    y_{m,n}(x,\sigma)=\hat{\psi}(x,i \sigma+m\lambda+i n\omega)
\end{equation}
where $i\sigma \in \{z: 0\leq\Im z<\omega, 0\leq\Re z<\lambda \}$.

\begin{Remark} \label{rmk:changeOfNotion}
{\em  Since the $p$ plane equation
only links values of $p$
differing by $m\lambda+i n\omega$, $m,n\in\ZZ$,
it is useful to think of  functions of $p$  as
vectors  with components  $m$ and $n$, parameterized by $\sigma$.}
\end{Remark}
\z Thus we rewrite  (\ref{eq:eqa}) as
\begin{multline} \label{eq:inhomo2}
\left(\frac{\partial^2}{\partial x^2} - \sigma -n \omega +i m\lambda \right)y_{m,n}\\
=\ i\psi_0 - 2\delta(x)y_{m,n} + \delta(x)\left(y_{m+1,n+1}+y_{m+1,n-1}\right)
\end{multline}
When $|n|+|m|\neq0$, the resolvent of the operator
$$-\frac{\partial^2} {\partial x^2} + \sigma + n \omega -i m \lambda$$ has the integral representation
\begin{equation} \label {eq:defOfg}
    \Big(\mathfrak{g}_{m,n}f\Big)(x) := \int_{\RR} G(\kappa_{m,n}(x-x')) f(x')dx'
\end{equation} with
\begin{displaymath} \kappa_{m,n}=\sqrt{-ip}=\sqrt{\sigma+n\omega-im\lambda}\
  \end{displaymath} where the choice of branch is so that
 if $p\in\mathbb{H}$, then $\kappa_{m,n}$ is in the fourth quadrant, and where the Green's function is given by
\begin{equation} \label{eq:defOfG}
    G(\kappa_{m,n}x)=\frac{1}{2}\kappa_{m,n}^{-1}e^{-\kappa_{m,n}|x|}
\end{equation}

\begin{Remark} \label{rmk:PropOfg}
     If $f(x)\in C_0^{\infty}$, using integration by parts we have, as $p\to \infty$
    $$\mathfrak{g}(f) \sim \frac{c(x)}{p} + o\left(\frac{1}{p}\right)$$
    where we regard $\mathfrak{g}$ as an operator with $p$ as a parameter; see also Remark \ref{rmk:changeOfNotion}. Furthermore, (\ref{eq:defOfg}) implies $c(x)\in L^{2}$.
\end{Remark}

Define the operator $\mathfrak{C}$ by
\begin{equation} \label{eq:definitionOfC}
    (\mathfrak{C}y)_{m,n}=\mathfrak{g}_{m,n}\left[ 2 \delta(x) y_{m,n} - \delta(x) \left(y_{m+1,n+1}+y_{m+1,n-1}\right) \right]
\end{equation}
Then Eq. (\ref{eq:inhomo2}) can be written in the equivalent integral form
\begin{equation} \label{eq:vector21}
    y=i\mathfrak{g}\psi_0 + \mathfrak{C}y
\end{equation}
where $\mathfrak{g}$ is  defined in  (\ref{eq:defOfg}).

\begin{Remark} \label{rmk:PropOfC}
  Because of the factor $\kappa_{m,n}^{-1}$ in (\ref{eq:defOfG}),
we have, with the identification in Remark \ref{rmk:changeOfNotion},
    $$\mathfrak{C}\phi(p) \sim \frac{c(x)}{\sqrt{p}} \phi(p) $$
as $p\to\infty$, for any function $\phi(p)$.
\end{Remark}

\subsection{Further transformations, functional space}

In this section we assume  $\psi_{0}\in C_{0}^{\infty}$.
As in Remark \ref{rmk:PropOfg},
 we obtain
\begin{equation} \label{eq:igpsi0}
    i\mathfrak{g}\psi_{0} = \frac{c_{1}(x)}{p} + O \left(\frac{1}{p^{3/2}} \right)
\end{equation}
for some $c_{1}(x)\in L^{2}$.

Let
\begin{equation} \label{eq:defOfh}
    h_{1}(p) = h_{1}(x,p) = c_{1}(x) \mathcal{L}\left(  {\bf 1}_{[0,1]}(t) \right)
\end{equation}
 For large $p$ we have
\begin{equation} \label{eq:asymptPropOfh}
    h_{1}(p) = \frac{c_{1}(x)}{p} + O \left(\frac{1}{p^{3/2}} \right)
\end{equation}

\begin{Remark} \label{rmk:invLaplaceOfh1}{\em
    As a function of $x$ , $h_{1}(p)$ is clearly in  $L^{2}$ and $$\mathcal{L}^{-1}(h_{1}(p)) = c_{1}(x) {\bf 1}_{[0,1]}(t)$$ thus for $t>1$ we have $$\mathcal{L}^{-1}(h_{1}(p)) = 0 $$}
\end{Remark}
\z Substituting
\begin{equation} \label{eq:changeOfVariable}
    y = y_{1} + h_{1}
\end{equation}
in (\ref{eq:vector21}) we have
\begin{equation} \label{eq:inNewVariable}
    y_{1} = i\mathfrak{g}\psi_{0} - h_{1} + \mathfrak{C}\left( h_{1} \right) + \mathfrak{C} y_{1}
\end{equation}
Let $y_{0}=i\mathfrak{g}\psi_{0} - h_{1} + \mathfrak{C}\left( h_{1} \right)$, then Remark \ref{rmk:PropOfC} implies that for large $p$
\begin{equation} \label{eq:asymptPropOfInhomo}
    y_{0} = O\left( \frac{1}{p^{3/2}} \right)
\end{equation}
and by construction $y_{0}\in L^{2}$ as a function of $x$.

 We analyze (\ref{eq:inNewVariable}) in the space $\mathscr{H}_b=L^2(\mathbb{Z}^2\times \mathbb{R},\|\cdot\|_b)$, $b>1$, where
\begin{equation} \label{eq:norm_in_y}
    \|y\|_b:=\left(\sum_{m,n}(1+|n|)^{\frac{3}{2}} e^{-b\sqrt{1+|m|}} \|y_{m,n}\|_{L^2}^2 \right)^\frac{1}{2}
\end{equation}
 We denote by $\hat{\psi}_{1}$ the transformed wave
function corresponding to $y_{1}$. Writing $y$ instead of $y_{1}$, we
obtain from (\ref{eq:inNewVariable}),
\begin{equation} \label{eq:vector2}
    y = y_{0} + \mathfrak{C} y
\end{equation}

\begin{Lemma}
\label{remark:C} $\mathfrak{C}$ is a compact operator on
$\mathscr{H}_b$,  and analytic in
$\sqrt{-ip}$. \end{Lemma}

\begin{proof} Compactness is clear since $\mathfrak{C}$ is a limit of bounded
finite rank operators. Analyticity is manifest in the expression
of $\mathfrak{C}$ (see (\ref{eq:defOfg}) and (\ref{eq:definitionOfC})).\end{proof}

\begin{Proposition} \label{prop:analyticityofy}
    Equation (\ref{eq:vector2}) has a unique solution iff the associated homogeneous equation
\begin{equation} \label{eq:vector3}
    y= \mathfrak{C}y
\end{equation}
has no nontrivial solution. In the latter case, the solution is analytic in $\sqrt{\sigma}$.
\end{Proposition}

\begin{proof}
    This follows from Lemma \ref{remark:C} and the Fredholm alternative.
\end{proof}

When $m=0$, $n=0$, and $\sigma=0$, $\mathfrak{C}$ is singular, but
the solution is not. Indeed,  by adding $\mathbf{1}_{[-A,A]}$,
$A>0$, to both sides of (\ref{eq:inhomo2}) we get the equivalent
equation
\begin{multline} \label{eq:inhomo2mod}
        \left( \frac{\partial^2}{\partial x^2} - \sigma -n \omega +i m\lambda +\mathbf{1}_{[-A,A]} \right)y_{m,n}\\
        =\ i\psi_0 +\left(\mathbf{1}_{[-A,A]}- 2\delta(x)\right)y_{m,n} + \delta(x)\left(y_{m+1,n+1}+y_{m+1,n-1}\right)
\end{multline}
Arguments similar to those when $\mathbf{1}_{[-A,A]}$ is absent show that the operator $\mathfrak{C}$ associated to (\ref{eq:inhomo2mod}) is analytic in $\sqrt{\sigma}$, thus $y_{m,n}$ is analytic in $\sqrt{\sigma}$.

\subsection{Equation for $A$}

Componentwise (\ref{eq:vector21}) reads
\begin{equation}\label{eq:detailed}
  \begin{split} y_{m,n} &=
    \int_{\RR}\frac{1}{2}\kappa_{m,n}^{-1}e^{-\kappa_{m,n}|x-x'|}
    \psi_0(x') dx' \\
    &+ \frac{1}{2\kappa_{m,n}} e^{-\kappa_{m,n} |x| } \left[ 2 y_{m,n}(0) -
      \left(y_{m+1,n+1}(0)+y_{m+1,n-1}(0)\right)\right] \end{split}
\end{equation}
With $A_{m,n}=y_{m,n}(0)$, we have
\begin{equation} \label{eq:inhomoInA}
    (\sqrt{\sigma + n\omega - im\lambda}-1)A_{m,n}
    = -\frac{1}{2} A_{m+1,n+1} - \frac{1}{2} A_{m+1,n-1} + g_{m,n}
\end{equation}
where $g_{m,n}$ is defined in (\ref{eq:defOfgmn}).

\begin{Proposition} \label{prop:yIsDeterminedByA}
The solution to (\ref{eq:detailed}) is determined by the $A_{m,n}$ through
\begin{multline}\label{link1}
    y_{m,n}
    =\int_{\RR} \frac{1}{2}\kappa_{m,n}^{-1}e^{-\kappa_{m,n}|x-x'|} \psi_0(x') dx'+ e^{-\kappa_{m,n} |x| } A_{m,n} - \frac{1}{\kappa_{m,n}} e^{-\kappa_{m,n} |x| } g_{m,n}
\end{multline}

It thus suffices to study (\ref{eq:inhomoInA}).
\end{Proposition}

\begin{proof}
 Taking $x=0$ in (\ref{eq:detailed}) we obtain (\ref{eq:inhomoInA}); using now (\ref{eq:inhomoInA}) in (\ref{eq:detailed}) we have
\begin{multline}
    y_{m,n}=\int_{\RR} \frac{1}{2}\kappa_{m,n}^{-1}e^{-\kappa_{m,n}|x-x'|} \psi_0(x') dx'\\
    + \frac{1}{2\kappa_{m,n}} e^{-\kappa_{m,n} |x| } \left[ 2 A_{m,n}-\left( A_{m+1, +1}+A_{m+1,n-1} \right) \right]\\
    =\int_{\RR} \frac{1}{2}\kappa_{m,n}^{-1}e^{-\kappa_{m,n}|x-x'|} \psi_0(x') dx'+ e^{-\kappa_{m,n} |x| } A_{m,n} - \frac{1}{\kappa_{m,n}} e^{-\kappa_{m,n} |x| } g_{m,n}
\end{multline}
\end{proof}

\begin{Remark} \label{rmk:spaceOfA}
If $y\in\mathscr{H}_{b}$, then $A_{m,n}=y_{m,n}(0)$ satisfies
(\ref{eq:norm_in_A_inhomo}).
\end{Remark}
\z Let $A^{0}_{m,n}=y^{0}_{m,n}(0)$ where $y^{0}_{m,n}$ is a solution to  (\ref{eq:vector3}). The solution of (\ref{eq:vector3}) has the freedom of a multiplicative constant; we choose it by imposing
\begin{equation} \label{eq:choiceOfA000}
    A^{0}_{0,0}=\lim_{\sigma\to 1} (\sigma-1)A_{0,0}
\end{equation}
\z It is clear $A^{0}_{m,n}$ satisfies the homogeneous equation associated to (\ref{eq:inhomoInA})
\begin{equation} \label{eq:homoInA}
    (\sqrt{\sigma + n\omega - im\lambda}-1)A^{0}_{m,n}
    = -\frac{1}{2} A^{0}_{m+1,n+1} - \frac{1}{2} A^{0}_{m+1,n-1}
\end{equation}

\subsection{Positions and residues of the poles} \label{subsec:positionsAndStrengthOfThePoles}
Define
\begin{equation} \label{eq:defOfsigma0}
    \sigma_{0}=1-\Big\lfloor\frac{1}{\omega}\Big\rfloor \omega
\end{equation}
To simplify notation we take $\omega>1$ in which case $\sigma_{0}=1$. The general case is very similar.

Denote
\begin{equation} \label{eq:curl_B}
    \iB:=\{in\omega+m\lambda+i: m\in\ZZ, n\in\ZZ, m\leq0, |n|\leq|m| \}
\end{equation}

\begin{Proposition} \label{prop:A} The system (\ref{eq:homoInA}) has nontrivial
solutions in $\mathscr{H}_b$ iff $\sigma=\sigma_{0}$(=1 as discussed above). If $\sigma=1$, then the solution is
a constant multiple of the vector $A_{m,n}^0$ given by
\begin{equation} \label{eq:norm_in_A_inhomoSplit}
    \left\{ \begin{split} A^0_{m,n}&=0 &m\geq 0 \ \mbox{and}\ (m,n)\neq (0,0)\\
    A^0_{m,n}&=0 &m\leq 0 \,\mbox{and}\ m\leq n\leq -m\\\
    A^0_{m,n}&=1 &(m,n)=(0,0)\end{split} \right.
\end{equation}
and obtained inductively from (\ref{eq:homoInA}) for all other $(m,n)$.
(Note that $\sigma=1$ is used
crucially here since (\ref{eq:homoInA}) allows for the nonzero value
of $A^0_{0,0}$.)
\end{Proposition}

\begin{proof} Let $\sigma=1$. By construction,
$A^{0}$ defined in Proposition~\ref{prop:A} satisfies the recurrence and we only need to check (\ref{eq:norm_in_A_inhomo}). Since
$$A^0_{m,n} = -\frac{A^0_{m+1,n+1}+A^0_{m+1,n-1}}{2(\sqrt{\sigma+n\omega+im\lambda}-1) } $$
and  $\sqrt{\sigma+n\omega+im\lambda}-1 \neq 0$, we have
$$|A^0_{m,n}| \leq C \frac{2^m}{\sqrt{(|n|+|m|)!}}$$
proving the claim.

Now, for any $\sigma$, if there exists a nontrivial solution, then
for some $n_0,m_0$ we have $A^0_{n_0,m_0}\neq 0$. By
(\ref{eq:homoInA}), we have either
\begin{equation} \label{inequ:1}
    |A^0_{n_0-1,m_0+1}|\geq \frac{1}{2}\,|(\sqrt{\sigma+n_0\omega+im_0\lambda}-1)|\cdot |A^0_{n_0,m_0}|
\end{equation}
or
\begin{equation} \label{inequ:2}
    |A^0_{n_0+1,m_0+1}|\geq \frac{1}{2}\,|(\sqrt{\sigma+n_0\omega+im_0\lambda}-1)|\cdot |A^0_{n_0,m_0}|
\end{equation}

It is easy to see that if $in_0\omega+m_0\lambda+i\in\iB^c$ or
$\sigma\neq1$, the above inequalities lead to
\begin{equation}
  \label{eq:grwth}
  |A^0_{n,m_0+m}| \ge  c\,\sqrt{m!}
\end{equation}
for large $m>0$ (note that in these cases $\sqrt{\sigma+n\omega+i
m\lambda}-1\neq 0$), contradicting (\ref{eq:norm_in_A_inhomo}).

Finally, if $\sigma=1$, then $A^0$ is determined by $A^0_{0,0}$ via
the recurrence relation (\ref{eq:homoInA}) (note that
$A^0|_{\iB^c}=0$). This proves uniqueness (up to a constant
multiple) of the solution.
\end{proof}



\z Combining Proposition \ref{prop:analyticityofy} and Proposition \ref{prop:A} we obtain the following result.
\begin{Proposition}
    The solution $\hat{\psi}(p)$ to equation (\ref{eq:inhomo1}) is analytic with respect to $\sqrt{-ip}$, except for poles in $\iB$.
\end{Proposition}

\begin{proof}
    Proposition \ref{prop:A} shows that (\ref{eq:homoInA}) has a solution $A^{0}$ for $\sigma \in \iB$;  by Proposition \ref{prop:analyticityofy}, $A$ has singularities in $\iB$, and the conclusion follows from Proposition
    \ref{prop:yIsDeterminedByA}.
\end{proof}

So far we showed that the solution has possible singularities in $\iB$. To show that indeed $\hat\psi$ has poles for generic initial conditions, we need the following result:

\begin{Lemma} \label{lemma:betterFredholm}Let $H$ be a Hilbert space. Let $K(\sigma):H\rightarrow H$ be compact, analytic in $\sigma$ and invertible in $B(0,r)\setminus\{0\}$ for some $r>0$. Let $v_0(\sigma)\notin \rm{Ran}(I-K(0))$ be analytic in $\sigma$. If $v(\sigma)\in H$ solves the equation $(I-K(\sigma))v(\sigma)=v_0(\sigma)$, then $v(\sigma)$ is analytic in $\sigma$ in $B(0,r)\setminus\{0\}$ but singular at $\sigma=0$. \end{Lemma}

\begin{proof} By the Fredholm alternative, $v(\sigma)$ is analytic when $\sigma\neq0$. If $v(\sigma)$ is analytic at $\sigma=0$ then $v_0$ is analytic and $v_0(\sigma)\in \rm{Ran}(I-K(0))$ which is a contradiction.
\end{proof}

The operator $\mathfrak{C}$ is compact by Remark \ref{remark:C}. The
inhomogeneity $y_0$ in equation (\ref{eq:vector3}) is analytic in
$\sqrt{\sigma}$. Furthermore, at $\sigma=1$,
$\rm{Ran}(I-\mathfrak{C})$ is of codimension 1 (Proposition
\ref{prop:A}). Combining with Lemma \ref{lemma:betterFredholm} we
have
\begin{Corollary} For a generic inhomogeneity $y_{0}$, $y(\sigma)$ is singular at $\sigma=1$.
Equivalently, $\hat{\psi}(p)$ has a pole at $p=i$.
\end{Corollary}

It can be shown that $\hat{\psi}(p)$ has a pole at $p=i$ for generic $\psi_0$. We prefer to show the following result which has a shorter proof.

\begin{Proposition} \label{prop:psiHasAPole}
    The residue $R_{0,0}$ of the pole for $\hat{\psi}$ at $p=i$ is given by
    \begin{equation} \label{eq:R00}
        R_{0,0} = \lim_{\sigma\to 1} (\sigma-1)A_{0,0}
        =\left[ -A_{1,-1}-A_{1,1}+2g_{0,0}\right]_{\sigma=1}
    \end{equation}
    In particular, $R_{0,0}\neq 0$ for large $\lambda$ and generic initial condition
    $\psi_{0}$.
\end{Proposition}

\begin{proof}
    When $m=0$ and $n=0$ (\ref{eq:inhomoInA}) gives
    \begin{equation} \label{eq:proofUse1}
        (\sqrt{\sigma}-1)A_{0,0} = -\frac{1}{2}A_{1,-1}-\frac{1}{2}A_{1,1}+g_{0,0}
    \end{equation}
    Clearly $A_{0,0}$ is singular as $\sigma \to 1$, which implies that $\hat{\psi}$ has a pole at $p=i$ with residue given in (\ref{eq:R00}). Thus $R_{0,0}$ is not zero if the quantity $\left[ -A_{1,-1}-A_{1,1}+2g_{0,0}\right]_{\sigma=1}$ is not zero.
    First, $g_{0,0}|_{\sigma=1}$ is not zero by definition:
    $$g_{0,0}|_{\sigma=1} = i\, \int_{\RR} \frac{1}{2} e^{-|x'|} \psi_0(x')dx'$$

    Next, taking $m=1$, $n=1$, and $\sigma=1$ in (\ref{eq:inhomoInA}) we obtain
    $$(\sqrt{1+\omega-i\lambda}-1)A_{1,1} = \left[-\frac{1}{2}A_{2,2}-\frac{1}{2}A_{2,0}+g_{1,1}\right]_{\sigma=1}$$
    Thus for any $c>0$ when $\lambda$ is large enough we have
    $$|A_{1,1}|\le c^{-1}\left[|g_{1,1}|+\max\{|A_{2,2}|,|A_{2,0}|\}\right]_{\sigma=1}$$
  Estimating  similarly  $A_{2,2}$ and $A_{2,0}$ and so on, we
see that $|A_{1,1}|=O(c^{-1})$. When $c$ is large enough we have $|A_{1,1}|<|g_{0,0}|$. Analogous bounds hold for  $A_{1,-1}$ showing that $\left[ -A_{1,-1}-A_{1,1}+2g_{0,0}\right]_{\sigma=1}$ is not zero.

\end{proof}

\begin{Corollary} \label{cor:psiHasPolesInB}
    For generic initial condition $\hat{\psi}$ has simple poles in $\iB$, and their  residues are given by $R_{m,n}=A^{0}_{m,n}$.
\end{Corollary}

\begin{proof}
    Take a small loop around $\sigma=0$ and integrate equation (\ref{eq:detailed}) along it. This gives a relation among $R_{m,n}$ which is identical to (\ref{eq:homoInA}):
    \begin{equation}
        (\sqrt{\sigma + n\omega - im\lambda}-1)R_{m,n}
        = -\frac{1}{2} R_{m+1,n+1} - \frac{1}{2} R_{m+1,n-1}
    \end{equation}
    Proposition \ref{prop:psiHasAPole} and (\ref{eq:choiceOfA000}) implies that $R_{0,0}=A^{0}_{0,0}$. The rest of the proof is follows from Proposition \ref{prop:A}.
\end{proof}

\begin{Remark} It is easy to see that there exist initial conditions for which the solution has no poles. Indeed, if the solutions $\psi_1$ and $\psi_2$  have a simple pole at $p=i$ with residue $a_1$ and $a_2$ respectively, then for initial condition $\psi_{0,0}=a_2\psi_{1,0}-a_1\psi_{2,0}$, the corresponding solution $\psi_{0}$ has no pole at $p=i$.
\end{Remark}

\subsection{Infinite sum representation of $A_{m,n}$}  \label{subsec:sumRepOfAmn}
Taking $\sigma=1$ in   (\ref{eq:inhomoInA}) we get
\begin{equation} \label{eq:inhomoAsigmaIs1}
    (\sqrt{1+n\omega-im\lambda}-1)A_{m,n}=-\frac{1}{2}A_{m+1,n-1}-\frac{1}{2}A_{m+1,n+1}+g_{m,n}
\end{equation}

For $\tau=(a_1,...,a_N)\in \{-1,1\}^N$, we define
$\tau_j^0=(a_1,...,a_j,0,...,0)$. (Note that $\tau=\tau^0_N$). We
denote $\Sigma \tau^0_j=\sum_{i=1}^ja_i$ and $\{-1,1\}^0=\{0\}$.

Let
$$B_{m,n}=\frac{1}{\sqrt{1+n\omega-im\lambda}-1}$$ and for some $\tau \in \{-1,1\}^{N}$ define
$$B_{m\,n\,N} = B_{m\,n\,N}(\tau) = \prod_{j=0}^{N-1}B_{m+j,\,n+\sum\tau^0_j}$$
Equation (\ref{eq:inhomoAsigmaIs1}) implies

\begin{equation} \label{eq:inhomogeneousMainmn} \begin{split}
A_{m,n}=(-1)^N \frac{1}{2^N} \sum_{\tau\in\{-1,1\}^N}B_{m\,n\,N-1}\,A_{m+N,n+\sum\tau}\\
+\sum_{j=0}^{N-1}(-1)^j \frac{1}{2^j}\sum_{\tau\in\{-1,1\}^j}B_{m\,n\,j}\,g_{m+j,\,n+\sum\tau}\\
\end{split} \end{equation}

 As $N\rightarrow\infty$ we have
$$\prod_{j=0}^{N} B_{m+j,n}\sim \frac{1}{\sqrt{N!}}$$ and
$A_{m,n}$ goes to zero as $m\rightarrow\infty$, and thus we have
$$\lim_{N\rightarrow\infty}(-1)^N \frac{1}{2^N} \sum_{\tau\in\{-1,1\}^N}B_{m\,n\,N-1}\,A_{m+N,n+\sum\tau}=0$$
In  the limit $N\rightarrow\infty$ we obtain
\begin{equation}\label{eq:sumRep}
A_{m,n}=\sum_{j=0}^{\infty}(-1)^j
\frac{1}{2^j}\sum_{\tau\in\{-1,1\}^i}B_{m\,n\,j}\,g_{m+j,\,n+\sum\tau}\end{equation}

\begin{Remark} \label{remark:errorControl}
    Truncating the infinite
expansion to $N$, the error is bounded by
    \begin{equation} \label{eq:errorControl}
        \left| \frac{1}{2^N} \sum_{\tau\in\{-1,1\}^N}(\prod_{j=0}^{N}B_{m+j,\,n+\sum\tau^0_j}\,A_{m+N,n+\sum\tau}^0\big) \right|
    \end{equation}
\end{Remark}

\section{Proof of Theorem \ref{thm:1}} \label{sec:proofOfThm1}

In \S \ref{subsec:positionsAndStrengthOfThePoles} it was shown that for a generic initial condition $\psi_{0}(x)$, the solution $\hat{\psi}(x,p)$ has simple poles in $\iB$, with residues $R_{m,n}=A_{m,n}$.

Since $y\in\mathscr{H}_{b}$, the inverse Laplace transform can be expressed using Bromwich contour formula. Recall that $y$ differs from the original vector form of $\hat{\psi}$ by (\ref{eq:changeOfVariable}), we have
\begin{equation} \label{eq:inverse_transform_limit}
    \psi(x,t)=\mathcal{L}^{-1}\hat{\psi}(x,p)= \mathcal{L}^{-1}(h_{1}) +\frac{1}{2\pi i} \int_{c-i\infty}^{c+i\infty} e^{p\,t} \hat{\psi}_{1}(x,p) dp
\end{equation}

The fact that $y\in\mathscr{H}_{b}$ also implies that
$\hat{\psi}_{1}(x,p) \rightarrow 0$ fast enough as $p\to c\pm
i\infty$. Thus the contour of integration in the inverse Laplace
transform can be pushed into the left half $p$-plane, after collecting the residues. As a result, for some small $c<0$ the contour becomes one
coming from $c-i\infty$, joining $c-i\epsilon$, $0$, and
$c+i\epsilon$ (for arbitrarily small $\epsilon>0$) in this order,
then going towards $c+i\infty$.

Thus we have
\begin{equation} \label{eq:psiDecomposition}
    \begin{aligned}
        \psi(t,x) = \mathcal{L}^{-1}(h_{1}) + \mbox{Res}|_{p=i}(\,e^{p\,t} \hat{\psi}_{1}) \\
        + \frac{1}{2\pi i} e^{c t} \int_{\epsilon}^{\infty} e^{i\,s\,t} \left(\hat{\psi}_{1}(x,c+i\,s)+\hat{\psi}_{1}(x,c-i\,s) \right) ds \\
        +\frac{1}{2\pi i} \int_{0}^{c-i\epsilon}e^{p\,t} \hat{\psi}_{1}(x,p)dp +\frac{1}{2\pi i}\int_{0}^{c+i\epsilon}e^{p\,t} \hat{\psi}_{1}(x,p)dp\\
    \end{aligned}
\end{equation}

By Corollary \ref{cor:psiHasPolesInB} we have
$$\mbox{Res}|_{p=i}(\,e^{p\,t} \hat{\psi}_{1}) = R_{0,0} = A^{0}_{0,0}$$

The third term in (\ref{eq:psiDecomposition}) decays exponentially
for large $t$ (since the integral is bounded), while the last two
terms yield an asymptotic power series in $1/\sqrt{t}$, as easily
seen from Watson's Lemma.

Combining these results and the fact $\mathcal{L}^{-1}(h_{1}) = o(1/t)$ (Remark \ref{rmk:invLaplaceOfh1}) concludes the first part of Theorem \ref{thm:1}, with $r(\lambda,\omega)=R_{0,0}$. The rest follows from Proposition \ref{prop:psiHasAPole}.

\section{Proof of Theorem \ref{thm:3}} \label{sec:proofOfThm3}

When $\omega=0$, the equation $$i\, \frac{\partial\psi }{\partial t}\,
=\Big(-\frac{\partial^2}{\partial x^2}-2\delta(x) + 2\delta(x)\,e^{-\lambda
  t}\cos(\omega t)\Big)\, \psi$$ becomes $$i\, \frac{\partial\psi
}{\partial t}\, =\Big(-\frac{\partial^2}{\partial x^2}-2\delta(x) +
2\delta(x)\,e^{-\lambda t}\Big)\, \psi$$

Rewriting  $A_{m,n}$ and $g_{m,n}$ as $A_{n}$ and $g_{n}$, (\ref{eq:inhomoInA}) becomes
$$(\sqrt{\sigma-im\lambda}-1)A_n = -A_{n+1} + g_n$$

Since $\omega=0$, (\ref{eq:sumRep}) simplifies to
\begin{equation} \label{eq:inhomogeneousMainn}
    A_{n}=\sum_{l=0}^{\infty}(-1)^{l-1} \prod_{j=0}^l \frac{1}{\sqrt{1-i(n+j)\lambda}-1} g_{n+l}
\end{equation}

\subsection{Proof of Theorem \ref{thm:3}, (i)}
When $n=1$ (\ref{eq:inhomogeneousMainn}) becomes
\begin{equation} \label{eq:inhomogeneousMain1}
A_{1}=\sum_{k=1}^{\infty}(-1)^k \prod_{j=1}^k \frac{1}{\sqrt{1-i\,j \lambda}-1} g_{k}
\end{equation}

With the notation
$$h_k=\prod_{j=1}^k \big( \sqrt{1-i\,j \lambda}-1\big)$$
(\ref{eq:inhomogeneousMain1}) becomes
$$A_1 = \sum_{k=1}^{\infty} \frac{(-1)^k\,g_k}{h_k}$$
Let
$$h_k= e^{w_k} \,\sqrt{\lambda^{k-1}(k-1)!}$$
We have
$$w_{k+1} - w_k = \log( \sqrt{1-i k \lambda} -1 ) - \frac{1}{2}
\log (\lambda k)$$
Differentiating in $\lambda$ we obtain
$$\frac{d}{d\lambda} \big( w_{k+1} - w_k \big) = \frac{1}{2 \lambda\sqrt{1-i\,k\,\lambda}}$$
Let $u_{k}$ be so that
$$\frac{d}{d\lambda}u_k = \frac{d}{d\lambda}w_k - \frac{i\,\sqrt{-i\,k}}{\lambda^{3/2}}$$
Then,
\begin{equation} \label{eq:diff_in_u}
\frac{d}{d\lambda}(u_{k+1}-u_k) = - \frac{i\,\sqrt{-i\,k -
i}}{\lambda^{3/2}} + \frac{i\,\sqrt{-i\,k}}{\lambda^{3/2}} +
\frac{1}{2 \lambda\sqrt{1-i\,k\,\lambda}}
\end{equation}

By taking the inverse Laplace transform of (\ref{eq:diff_in_u}) in $k$ we get (we use $p$ as the transformed variable here)
\begin{equation}
  \label{eq:add1}
  (e^{-p}-1)\,\mathcal{L}^{-1} \frac{d}{d\lambda}u_k = \frac{\sqrt{i\,}(1-e^{-p}+e^{-ip/\lambda}\,p)} {2\,\sqrt{\pi}\,(p\lambda)^{3/2}}
\end{equation}
Integrating (\ref{eq:add1}) with respect to $\lambda$ gives
$$\mathcal{L}^{-1} u_k = \frac{\sqrt{i\,} \lambda \left( -2 + 2 e^{-p} - i^{3/2}\sqrt{p\,\pi\,\lambda}\,\rm{erf}\,(\frac{-i^{3/2}\sqrt{p}}{\sqrt{\lambda}}) \right)}{2\,(-1+e^{-p})\,\sqrt{\pi}\,(p\lambda)^{3/2}}$$

Thus
\begin{displaymath} \begin{split}
    \frac{1}{h_k}=& \frac{e^{-w_k}}{\sqrt{\lambda^{k-1}(k-1)!}} \\
    =& \exp \Big( -\int_{0}^{\infty} e^{-kp} \frac{\sqrt{i\,}\lambda \left( -2 + 2 e^{-p} -i^{3/2}\sqrt{p\,\pi\, \lambda}\,\mbox{erf}(\frac{-i^{3/2}\sqrt{p}} {\sqrt{\lambda}}) \right)} {2\,(-1+e^{-p})\,\sqrt{\pi}\,(p \lambda)^{3/2}}dp\Big)\\
    &\times\exp \Big( -\frac{2\,i\,\sqrt{-i\,k}}{\sqrt{\lambda}} \Big) \lambda^{\frac{1-k}{2}}\frac{1} {\sqrt{\Gamma(k)}}\\
\end{split} \end{displaymath}

Finally we obtain
\begin{displaymath} \begin{aligned}
A_1 &= \sum_{k=1}^{\infty}\frac{(-1)^k g_k}{h_k} = \mathcal{L}\sum_{k=1}^{\infty} (-1)^k \mathcal{L}^{-1}\left(\frac{g_k}{h_k}\right) \\
&= \int_{0}^{\infty} \sum_{k=1}^{\infty} (-1)^k e^{-kp} \mathcal{L}^{-1}\left(\frac{g_k}{h_k}\right) dp
= \int_{0}^{\infty} \frac{-e^{-p}}{1+e^{-p}} \mathcal{L}^{-1} \left(\frac{g_k}{h_k}\right) dp \\
&= \int_{0}^{\infty} \frac{-e^{-p}}{1+e^{-p}} \int_{c-i\infty}^{c+i\infty}g_k \exp \Big( -\frac{2\,i\,\sqrt{-i\,k}}{\sqrt{\lambda}} \Big)\lambda^{\frac{1-k}{2}} \frac{1}{\sqrt{\Gamma(k)}}\cdot \\
&\exp \Big( -\int_{0}^{\infty} e^{-kp} \frac{\sqrt{i\,} \lambda \left( -2+2 e^{-p}-i^{3/2}\sqrt{p\,\pi\,\lambda} \,\mbox{erf}(\frac{-i^{3/2}\sqrt{p}} {\sqrt{\lambda}}) \right)}{2\,(-1+e^{-p})\,\sqrt{\pi}\,(p\lambda)^{3/2}} dp \Big) dk dp\\
\end{aligned} \end{displaymath}

\subsection{Proof of Theorem \ref{thm:3}, (ii)} \label {sec:AsymptBehaviorOmega0}
Here we assume that  expansion of $A_{1}$ as $\lambda\to 0$ is invariant under
a $\frac{\pi}{2}$ rotation; that is, there are no Stokes lines
in the fourth quadrant; this would be ensured by Borel summability
of the expansion in $\lambda$.

Let $\lambda=ir$ with $r<0$, and for simplicity let $g\equiv1$, then (\ref{eq:inhomogeneousMain1}) implies

\begin{equation}\label{eq:inhomogeneousMain1_special}\begin{split}
A_1 &= \sum_{n=1}^{\infty} \prod_{k=1}^{n}\frac{(-1)^k}{\sqrt{1+kr}-1} = \sum_{n=1}^{\infty}\frac{\prod_{k=1}^{n}(\sqrt{1+kr} + 1)}{(-1)^k k! r^k}\\
&= \sum_{n=1}^{\infty} \frac{\exp(\sum_{k=1}^{n}\log(\sqrt{1+kr}+1))} {(-1)^k k! r^k}\\
\end{split}\end{equation}
The Euler-Maclaurin summation formula gives
\begin{equation}\label{eq:E_M_1}
    \begin{split}
        \sum_{k=1}^n \log(\sqrt{1+kr}+1) \sim \int_0^n \log(\sqrt{1+xr}+1)dx +C\\
        =\,-\frac{1}{r} + k\log(\sqrt{1+kr}+1) -\frac{1}{2}k+ \frac{\sqrt{1+kr}}{r}+C\\
    \end{split}
\end{equation}
where
$$C \sim \sum_{k=1}^{1/r} \log(\sqrt{1+kr}+1) - \int_0^{1/r}\log(\sqrt{1+xr}+1)dx \sim -\frac{\log(2)}{2}$$

Therefore
\begin{equation} \label{equ:A1_middle}
    A_1 \sim \sum_{k=1}^{\infty}\frac{\exp(-\frac{1}{r} + k\log(\sqrt{1+kr}+1) -\frac{1}{2}k +\frac{\sqrt{1+kr}}{r})}{(-1)^k k! r^k}
\end{equation}

Since
\begin{displaymath}
    \begin{split}
        &\frac{\exp(-\frac{1}{r}+k\log(\sqrt{1+kr}+1) -\frac{1}{2}k + \frac{\sqrt{1+kr}}{r})}{(-1)^k k! r^k} \\
        \sim& \exp\left({-\frac{3}{2r} + \frac{2}{3}\sqrt{r}(k+\frac{1}{r})^{(3/2)} -\frac{\log(2)}{2}-\frac{\log(\pi)}{2}+\frac{\log(-r)}{2}}\right)\\
    \end{split}
\end{displaymath}
applying the Euler-Maclaurin summation formula again gives
\begin{equation} \label{eq:A1Final}
    A_1 \sim \frac{2^{1/3}3^{1/6} \Gamma(\frac{2}{3})e^{-\frac{3i}{2\lambda}}(-i\lambda)^{1/6}}{2\sqrt{\pi}}
\end{equation}

\subsection{Numerical results} \label{subsec:numericalEvidence} Figure
\ref{fig:2} shows $\log(|R_0|)$ as a function of $\log(\lambda)$, very
nearly a straight line  with slope $1/6$ (corresponding to the $\lambda^{1/6}$
behavior), with good accuracy good even until $\lambda$ becomes as
large as 1.

\begin{figure}[ht!]
\includegraphics[scale=0.4]{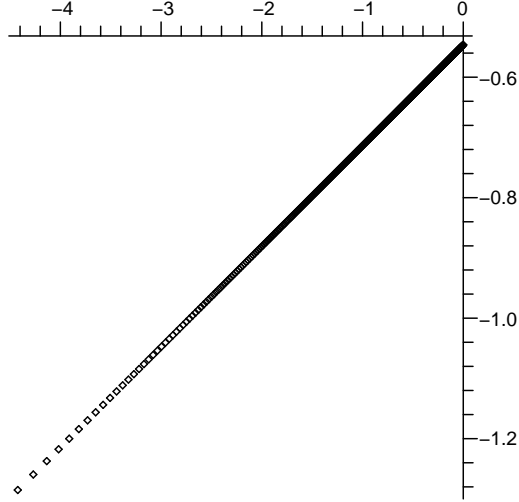}
\caption{Log-log plot of $|R_{0}|$ as a function of $\lambda$ for
  $\omega=0$. $R_{0,0}$ is the residue of the pole of
  $\hat{\psi}(p,x)$ at $p=i$, see Corollary \ref{cor:psiHasPolesInB}.}\label{fig:2}
\end{figure}

\section{Ionization rate under a short pulse} \label{sec:shortPulse}
We now consider a short pulse, with fixed total energy and fixed total number of oscillations.  The corresponding Schr\"odinger equation is
\begin{equation} \label{eq:shroedinger_short_pulse}
    i\, \frac{\partial\psi }{\partial t}\,=\Big(-\frac{\partial^2}{\partial x^2}-2\delta(x)+2\lambda\,\delta(x)\,e^{-\lambda t} \cos(\omega t)\Big)\, \psi
\end{equation}
where $\lambda$ is now a large real parameter (note the $\lambda$ in front of the exponential). We are interested in the ionization rate as $\lambda \rightarrow \infty$.

By similar arguments as in \S \ref{subsec:sumRepOfAmn} we have the convergent representation
\begin{equation} \label{eq:inhomogeneousMainmn_short_pulse}
    A_{m,n} = \sum_{i=0}^{\infty}(-1)^i \left(\frac{\lambda}{2}\right)^i \sum_{\tau\in2^i}\prod_{j=0}^i
B_{n+j,\,m+|\tau^0_j|}g_{n+i,\,m+|\tau|}
\end{equation}
Figs.  \ref{fig:5} and \ref{fig:6} give  $|R_{0,0}|$ (see Corollary \ref{cor:psiHasPolesInB}) and $\lambda$ under different $\omega/\lambda$ ratios. On  small scales on the $\lambda$ axis,  $|R_{0,0}|$  exhibits rapid oscillations,
easier seen  if $\omega/\lambda$ is smaller.
\begin{figure}[ht!]
\includegraphics[scale=0.6]{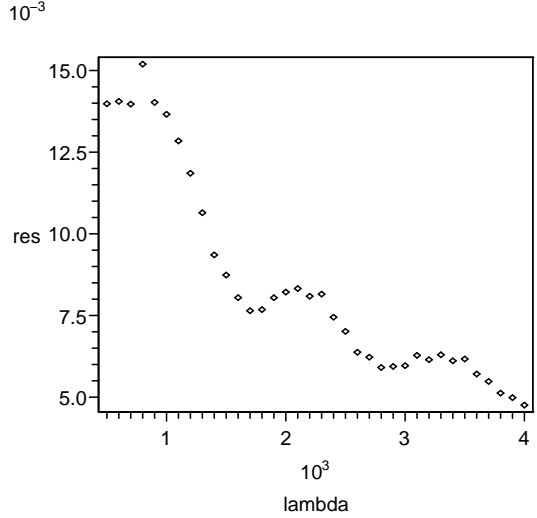}
\caption{$|R_{0,0}|$ as a function of $\lambda$, with fixed ratio $\omega/\lambda=20$.  $R_{0,0}$ is the residue of the pole of $\hat{\psi}(p,x)$ at $p=i$, see Corollary \ref{cor:psiHasPolesInB}.}\label{fig:5}
\end{figure}
\begin{figure}[ht!]
\includegraphics[scale=0.6]{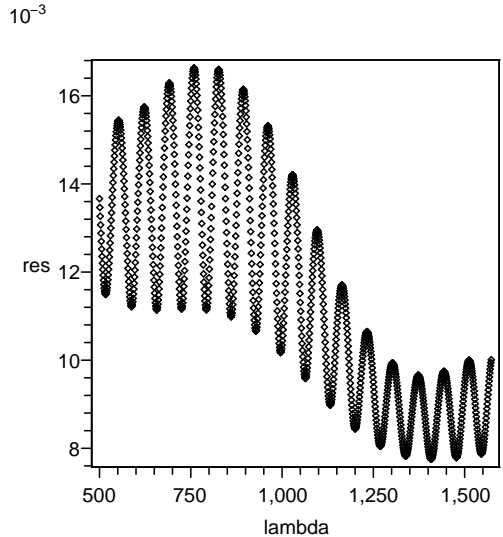}
\caption{$|R_{0,0}|$ as a function of $\lambda$ when $\omega/\lambda=5$. $R_{0,0}$ is the residue of the pole of $\hat{\psi}(p,x)$ at $p=i$, see Corollary \ref{cor:psiHasPolesInB}.}\label{fig:6}
\end{figure}

\section{Results for $\lambda=0$ and $\omega\neq0$} \label{sec:anotherSpecialCase}
We briefly go over the case $\lambda=0$, where ionization is complete; the full analysis is done in \cite{CRM}.
In this case $\hat{\psi}$ does not have poles on the imaginary line; we give a summary of the argument in \cite{CRM}.

The homogeneous equation now reads
\begin{equation} \label{eq:A01}
    \sqrt{\sigma + m\omega}\, A_m = -\frac{1}{2}A_{m+1} -\frac{1}{2}A_{m-1} + A_m
\end{equation}
Thus we have
$$\sum_{\NN}\sqrt{\sigma + m\omega}\, A_m \overline{A_m} =
-\frac{1}{2} \sum_{\NN} A_{m+1} \overline{A_m} - \frac{1}{2} \sum_{\NN} A_{m-1} \overline{A_m} + \sum_{\NN} A_m \overline{A_m}$$

The first sum and the second sum on the right hand side are conjugate to each other, and each term in the third sum is real. So the right hand side is real, thus the left hand side is also real.

For $\Im(\sigma)\neq 0$, $\Im\left(\sqrt{\sigma + m\omega}\, A_m \overline{A_m}\right)$ has same the sign as $\Im\sigma$. Therefore the sum can not be real and the equation has no nontrivial solution. When $\Im(\sigma)=0$, for $m<0$, all $\Im\left(\sqrt{\sigma+m\omega}\,A_m \overline{A_m}\right)$ have the same sign and  for $m\geq 0$, $\sqrt{\sigma+m\omega}\,A_m \overline{A_m}$ is real. Since
the final sum is purely real, this means $A_m=0$ for $m<0$. But then, recursively,  all $A_m$ should be 0.

Zero is thus the only solution to (\ref{eq:A01}).
 By the Fredholm alternative the solution $A$ is analytic in $\sqrt{\sigma}$ and thus the associated $y$  is analytic in $\sqrt{\sigma}$. This entails complete ionization.

\begin{figure}[ht!]
\includegraphics[scale=0.5]{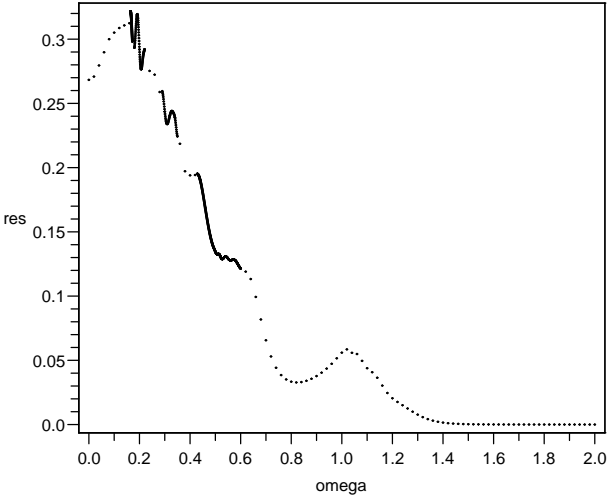}
\caption{$|R_{0,0}|$, at $\lambda=0.01$, as a function of $\omega$.  $R_{0,0}$ is the residue of the pole of $\hat{\psi}(p,x)$ at $p=i$, see Corollary \ref{cor:psiHasPolesInB}.}\label{fig:7}
\end{figure}
\begin{figure}[ht!]
\epsfig{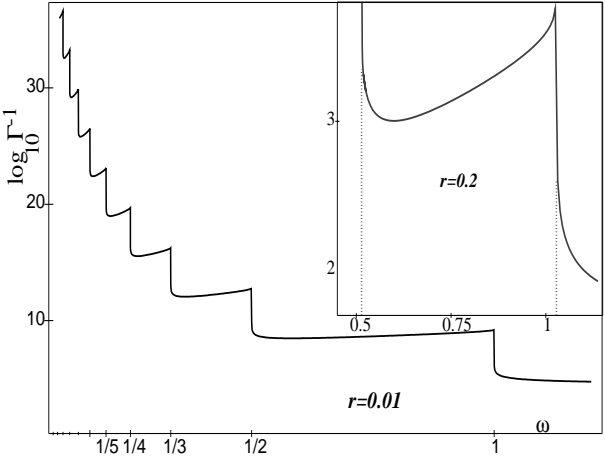}

\caption{$\log_{10}\Gamma^{-1}$, at $\lambda=0$,  as a function of $\omega/\omega_0$.  $\Gamma$
is the Fermi Golden Rule exponent for the probability decay
and $\hbar\omega_0=E_0$, the energy of the bound state of one delta function
of amplitude $r$ \cite{JPA}.}\label{fig:8}

\end{figure}

\subsection{Small $\lambda$ behavior.}
We expect that the behavior of the system at $\lambda=0$ is a limit of the one for small $\lambda$. However, this limit is very singular, as the density of the poles in the left half plane goes to infinity as $\lambda\to 0$, only to become finite for $\lambda=0$. Nonetheless, given a $\lambda$, small but not extremely small, formula (\ref{eq:sumRep}) allows us to calculate the residue.

Figure \ref{fig:7} shows the behavior of the residue versus $\omega$, for $\lambda=0.01$.

We show for comparison the corresponding result when $\lambda=0$, in Figure \ref{fig:8}:

\section{Acknowledgments.}
This work was supported in part by the National Science Foundation
 DMS-0601226 and  DMS-0600369.

\end{document}